\documentclass[letterpaper, 10 pt, conference]{ieeeconf}  

\IEEEoverridecommandlockouts                              

\usepackage{graphicx}      
\usepackage{amsmath,amssymb,amsfonts}

\usepackage{soul}
\usepackage{comment}
\usepackage[utf8]{inputenc} 
\usepackage{tikz}
\usepackage{algorithmic}
\usepackage{graphicx}
\usepackage{textcomp}
\usepackage{xcolor}
\usepackage{subfigure}
\usepackage{cite}
\usepackage{paralist}

\newtheorem{proposition}{Proposition}[section]

\newtheorem{remark}{Remark}[section]
\newtheorem{definition}{Definition}[section]
\newtheorem{theorem}{Theorem}[section]

\newtheorem{assumption}{Assumption}

\title{\LARGE \bf On feedback passivation under sampling}

\author{Mattia Mattioni$^1$, Alessio Moreschini$^{1, 2}$, Salvatore Monaco$^1$, Doroth\'ee Normand-Cyrot$^2$
\thanks{Supported by \emph{Universit\'e Franco-Italienne/Universit\`a Italo-Francese} (Vinci Grant 2019) and by \emph{Sapienza Universit\`a di Roma} (\emph{Progetti di Ateneo 2018}-Piccoli progetti \emph{RP11816436325B63}).}
\thanks{$^1$Dipartimento di Ingegneria Informatica, Automatica e Gestionale \emph{A. Ruberti} (Sapienza University of Rome); Via Ariosto 25, 00185 Rome, Italy {\tt\small {\{mattia.mattioni, alessio.moreschini, salvatore.monaco\}@uniroma1.it}}.
}%
\thanks{$^2$Laboratoire de Signaux et Syst\`emes (L2S, CNRS and \emph{Univ. Paris-Saclay}); 3, Rue Joliot Curie, 91192, Gif-sur-Yvette, France {\tt\small  \{alessio.moreschini,dorothee.normand-cyrot \}@centralesupelec.fr}}%
}

\begin{document}

\maketitle
\thispagestyle{empty}
\pagestyle{empty}

\begin{abstract} 
In this paper we show that feedback passivation under sampling can be preserved under digital control through the redefinition of a passifying output map which depends on the sampling period. The design is constructive and approximate solutions make sense. The procedure is applied to port Hamiltonian dynamics and Interconnection and Damping Assignment feedback. Performances are  illustrated over the gravity pendulum example.
\end{abstract}

\begin{keywords}
Sampled-data control, Stability of nonlinear systems, Computer-aided control design.
\end{keywords}

\section{Introduction}

Energy-Based (EB) approaches constitute a physically inspired powerful setting for the design of control systems (see \cite{ortega2001putting,ortega2002interconnection,hatanaka2015passivity} and references therein). The basic idea is to achieve stabilization to a desired equilibrium and to assign the transient behaviour via energy transfer between the interconnected parts of the dynamical system. In this sense, these strategies enlarge the stabilizing goal of standard passivity-based control (PBC) and can be seen as the second generation of PBC with broad applicability to various domains as discussed in \cite{hatanaka2015passivity} with reference to robotics. As passivity is the core of PBC, feedback passivation is the underlying leitmotiv of energy-based control at large. More precisely, feedback passivation relies on the design of a controller that shapes the dissipation according to a given target storage  (or energy-like) function while ensuring closed-loop passivity with respect to a suitably defined output. Feedback passivation is in fact instrumental in celebrated nonlinear control strategies such as  backstepping or feedforwarding regarding cascade dynamics \cite{sepulchre2012constructive} or energy balance and Interconnection and Damping Assignment (IDA-PBC) for Hamiltonian dynamics at large \cite{ortega2002interconnection,ortega2002stabilization,ortega2008control}.

All of this essentially concerns the continuous-time framework whereas existing results in discrete time are few. This essentially results from the fact that the property of passivity itself faces the difficulty of necessary input dependency of the output map and the generic nonlinearity of the functions involved in the  computation of control solutions. However, stimulated by an increasing interest toward computer-oriented designs, several passivity-based digital control strategies are developed in the discrete-time literature \cite{stramigioli2005sampled, aoues2017Hamiltonian}, for cascade or interconnected dynamics also, but with no universal agreement on these methods. The main reason seems to be the lack of a clear underlying link among these methods and of the instrumental definition of a feedback passivation procedure in discrete time. We propose to set a bridge between these continuous-time and discrete-time strategies through the understanding of feedback passivation under sampling. In this context, the present paper investigates the very immediate question: as soon as the assignment of a certain target energy function is involved in the continuous-time design, is it possible to preserve this goal under
digital control?
We recall that sampled-data dynamics, issued from the sampling of continuous-time ones under piecewise constant control, are very attractive because of their practical interest but also because endorsing some of the properties of the original continuous-time plant  \cite{monaco2010sampled}.

To discuss this question, we make reference to the notion of average passivity that we introduced in \cite{monaco2011nonlinear} to overcome the input-to-output obstruction and the concept of matching a target behaviour rather than the complete state dynamics under digital control and at all sampling times. In \cite{monaco2010sampled},  we showed that average passivity is recovered under sampling when making reference to a modified output map, that comes to depend on the sampling period. In \cite{tanasa2015backstepping}, we shown that digital Lyapunov-based design strategies (e.g., backstepping) can be worked out in terms of Input-Lyapunov matching at the sampling instants, through suitably computed digital control laws, that come out to depend on the sampling period too. The combined use of average passivity and input-Lyapunov matching was then further developed for digital PBC stabilization at the origin via output damping \cite{monaco2010sampled} and IDA-PBC of Hamiltonian dynamics \cite{tiefensee2010ida}.

 Hereinafter we show that whenever a dynamics is feedback-passive in continuous time, it is feedback average passive under sampling with the same target storage, but with respect to a modified  passifying output, explicitly depending on the control and the sampling period. The solution we provide is constructive for the digital feedback that is defined by its series expansion in powers of the sampling period,  around the continuous-time solutions. As exact solutions are seldom computable in practice,  approximate controllers can easily be defined by truncating the series solution at any desired order.  When applied to the class of port-controlled Hamiltonian (pcH) systems and IDA-PBC control, the proposed passivation design extends previous results in \cite{tiefensee2010ida}.
Stabilization of the simple gravity pendulum at any desired point is worked out to illustrate the result and the effectiveness of approximate feedback solutions with respect to standard emulation.

The paper is organized as follows. Preliminaries are given in  Section \ref{sec:recalls} and the problem is formally stated. The main result is in Section \ref{sec:main} and further specified in \ref{sec:pcH} for pcH systems. A simulated example is carried out in Section \ref{sec:ex}. Conclusions and perspectives are in Section \ref{sec:conc}.
\smallskip

\emph{Notations. } Functions and vector fields are assumed smooth and complete over the respective definition spaces. $\mathbb{R}$ and $\mathbb{N}$ denote the set of real and natural numbers including $0$. For any vector $z \in \mathbb{R}^n$, $\|z\|$ and $z^\top$ define respectively the norm and transpose of $z$.  $I_d$ and $\text{I}$ and denote respectively the identity matrix and identity operator. $\mathrm{L}_f=\sum_{i=1}^{n}f(\cdot)\frac{\partial}{\partial x_i}$ denotes the Lie derivative and $e^{\mathrm{L}_f}=I + \sum_{i\geq 1}\frac{\mathrm{L}_{f}^i}{i!}$ the exponential  Lie series operator, associated with the vector field $f$. Given two vector fields $f(x)$, $g(x)$, $ad_f g (x)= (\mathrm{L}_f \mathrm{L}_g - \mathrm{L}_g \mathrm{L}_f)(x)$ denotes their Lie bracket. Given a twice continuously differentiable function $S(\cdot):\mathbb{R}^n\to\mathbb{R}$, $\nabla S$ represents its gradient (column) vector and $\nabla^2 S$ its Hessian matrix.  A function $R(x,\delta)= O(\delta^p)$ is said of order $\delta^p$, $p \geq 1$ if whenever it is defined it can be written as $R(x, \delta) = \delta^{p-1}\tilde R(x, \delta)$ and there exist a function $\theta \in \mathcal{K}_{\infty}$ and $\delta^* >0$ s. t. $\forall \delta\leq \delta^*$, $| \tilde R (x, \delta)| \leq \theta(\delta)$. 

\section{Recalls and Problem Statement}\label{sec:recalls}

Basic results on passivity and feedback passivity for continuous-time  and discrete-time systems are recalled.

\subsection{Feedback passivation in continuous time}
Consider the input-affine continuous-time dynamics 
\begin{align}\label{dyn1}
\dot{x}&=f(x) + g(x)u
\end{align}
with $x \in \mathbb{R}^n $, $u \in \mathbb{R}$. Let $x_\star \in \mathbb{R}^n$,  a desired equilibrium in $\{ x \in \mathbb{R}^n \text{ s.t. } g^\perp(x) f(x) = 0 \}$, with $g^\perp(\cdot)$ the maximal rank annihilator of $g(\cdot)$ ($g^\perp(x)g(x)=0$ for all $x\in\mathbb{R}^n$), and assume the dynamics feedback passive   \cite{byrnes1991passivity,sepulchre2012constructive} according to the definition below. 

\begin{assumption}\label{As1}
Given \eqref{dyn1}, there exist smooth functions $\gamma(\cdot): \mathbb{R}^n \to \mathbb{R}$, $h_d(\cdot) : \mathbb{R}^n \to \mathbb{R}$ and $S_d(\cdot): \mathbb{R}^n \to \mathbb{R}_{\geq 0}$ such that the feedback law 
	\begin{align}\label{passfeed}
	u=\gamma(x)+v
	\end{align}
	makes the closed-loop system
	\begin{subequations}\label{dyn2}
	\begin{align}\label{dyn2a}
	\dot{x}&=f_d(x) + g(x)v\\\label{out2}
	y&=h_d(x)
	\end{align}
	\end{subequations}
passive with $f_d(x) := f(x) + g(x) \gamma(x)$ and storage function $S_d(\cdot)$; i.e.  dissipation  holds for all $t \geq 0$ and $x_0 \in \mathbb{R}^n$
	\begin{align}\label{pass1}
	S_d(x(t))-S_d(x_0)\leq\int_{0}^{t}y(s) v(s) \mathrm{d}s.
	\end{align}
\end{assumption}

\smallskip 
Once feedback passivation is guaranteed whenever $S_d(x_\star) = 0$ and (\ref{dyn2}) is zero-state detectable (ZSD), then the desired equilibrium $x_\star$ can be stabilized through damping injection setting $v = -\kappa h_d(x)$  with $\kappa >0$. From the dissipation inequality (\ref{pass1}) and the Kalman-Yakubovich-Popov (KYP) property \cite{byrnes1991passivity} one easily verifies 
\begin{align}
\dot S_d(x) = \underbrace{\mathrm{L}_{f_d} S_d(x)}_{\leq 0} + v \underbrace{\mathrm{L}_g S_d(x)}_{= h_d(x)} \leq - \kappa \| h_d(x)\|^2
\end{align}
and thus asymptotic stability of $x_\star$ for (\ref{dyn2}).
From now on, with no loss of generality, we set $h_d(x) = \mathrm{L}_g S_d(x)$.

\subsection{Discrete-time passivity and problem statement}
The notion of discrete-gradient is recalled  \cite{mclachlan1999geometric}.
\begin{definition}
	\label{def:1}
	Given a smooth real-valued function $S_d(\cdot):\mathbb{R}^n\to\mathbb{R}$, its discrete gradient is a vector-valued function of two variables, $\bar{\nabla}S_d|_{x}^{z}: \mathbb{R}^n \times \mathbb{R}^n \to \mathbb{R}^n$ satisfying for all $x,z \in \mathbb{R}^n$
	\begin{align*}
	S_d(z)-S_d(x)&=	(z-x)^\top\bar{\nabla}S_d|_{x}^{z}, \quad \bar{\nabla}S_d|_{x}^{x}={\nabla}S_d(x).
	\end{align*}
\end{definition}

\smallskip
By the mean-value theorem, the discrete gradient can be computed according to the integral form
\begin{align*}
\bar{\nabla} S_d |_{x}^{z}=\int_{0}^{1 }\nabla S_d(x + \ell(z-x) ) \mathrm{d}\ell
\end{align*}
so getting the following approximation 
\begin{align}\label{discrete:approx}
\bar{\nabla}S_d|_{x}^{z} \! = \! \nabla S_d(x)\! +\! \frac{1}{2}\! \nabla^2 S_d(x)(z\! -\! x)\! +\! O(\| z-x\|^2).
\end{align}

Consider now a discrete-time dynamics 
	\begin{align}
		x_{k+1}&=x_k + F(x_k, u_k)\label{disdyn}
	\end{align}
with $x \in \mathbb{R}^n$, $u \in \mathbb{R}$ and let $x_\star \in \{ x\in \mathbb{R} \text{ s.t. } \exists u\in \mathbb{R} : \ F(x, u) = 0  \}$, an equilibrium to be stabilized.

When considering an output map $h(\cdot,u): \mathbb{R}^n \rightarrow \mathbb{R}$ depending on $u$, the following definition of discrete-time passivityis given.

\begin{definition}
	The discrete-time dynamics (\ref{disdyn}) with output $Y(x,u)=h(x,u)$ is said \emph{passive} if there exists a smooth positive definite function $S_d(\cdot): \mathbb{R}^n \to \mathbb{R}_\geq 0$  (the storage function) verifying for all $k\geq 0$ and $x_0 \in \mathbb{R}^n$ 
	\begin{align*}
	S_d(x_{k})-S_d(x_0) \leq \sum_{j = 0}^{k-1}u_j Y(x_j, u_j)  
	\end{align*}
	or, equivalently, for all $(x_k, u_k) \in \mathbb{R}^n \times \mathbb{R}$,
	\begin{align*}
	\Delta_k S_d(x) : = S_d(x_{k+1})-S_d(x_k)  \leq u_k Y(x_k, u_k).
	\end{align*}
\end{definition}

\smallskip

\noindent The one step ahead increment of the storage function along the dynamics (\ref{disdyn}) can be rewritten in terms the discrete gradient function in Definition \ref{def:1} as
\begin{align}
& S_d(x + F(x, u)) - S_d(x)= \bar \nabla^\top S_d|_{x}^{x + F(x, u)}F(x, u)\label{dtd}
\end{align}
omitting the $k$-dependency when clear from the context. It can be further characterized according to (\ref{discrete:approx}) as
\begin{align*}
 &S_d(x + F(x, u)) - S_d(x)=\nabla^\top S_d(x)F(x,u) \\
 &+ F^\top (x,u)\frac{1}{2} \nabla^2 S_d(x) F(x,u) +  O(\|F(x,u) \|^2).
\end{align*}

Hereinafter, we consider discrete-time systems issued from sampling continuous-time dynamics of the form (\ref{dyn1}) under piecewise constant control and sampled measures of the state. More precisely, we set $u(t) = u_k = u(k\delta)$ for all $t \in [k\delta, (k+1)\delta[$, $k\geq 0$, $x_k = x(k\delta)$ and sampling period $\delta \in ]0, T^\star[$. Through usual integration, the so-called equivalent sampled-data model to  (\ref{dyn1}) can be described in the form of a map as in (\ref{disdyn}) so getting
\begin{align}\label{sdequiv}
x_{k+1} = x_k + F^\delta(x_k, u_k)
\end{align}
 with  $x\! +\! F^\delta(x,u) \!=\! e^{\delta (\mathrm{L}_f + u \mathrm{L}_g)}x\! =\! x + \sum_{i >0}\frac{\delta^i}{i!}(\mathrm{L}_{f} + u\mathrm{L}_g)^ix$. 

It is easily verified that (\ref{sdequiv}) preserves neither the input affine structure of (\ref{dyn1}) nor its properties in general. Among them, passivity of (\ref{dyn1}) under sampling when setting $y_k=h_d(x_k)$ is lost because of the lack of a direct input-output link. However, as proved in  \cite{monaco2010sampled} that, whenever (\ref{dyn1}) is passive in continuous time, passivity under sampling is preserved under a new output map $$h^\delta(x, u) = \frac{1}{u}\bar \nabla^\top S_d|_{x + F^\delta(x,0)}^{x + F^\delta(x, u)} \Big(F^\delta(x, u) - F^\delta(x,0) \Big).$$

It results that preservation of feedback passivity under sampling still represents a challenging problem because the passifying output must be changed. 
Does feedback passivity of the continuous-time dynamics (\ref{dyn1}) implies some feedback passivity of the sampled-data equivalent model (\ref{sdequiv})? This is the question discussed in the sequel. 

\section{Digital passivation and stabilization}  \label{sec:main}
 
Before stating the main result, the following Proposition is recalled from \cite{tanasa2015backstepping}. 
\begin{proposition}
Let the system (\ref{dyn1}) verify Assumption \ref{As1} with storage function $S_d(\cdot): \mathbb{R}^n \to \mathbb{R}_{\geq 0}$ verifying $\mathrm{L}_g S_d(x)\neq 0$ for all $x \neq x_\star$. Then, there exists $T^\star>0$ such that for all $\delta \in [0, T^\star[$ and $k\geq 0$, the Input-$S_d$-Matching (IS$_d$M ) equality 
\begin{align}\label{IS$_d$M}
S_d(x_k\! +\! F^\delta(x_k, u_k)) \! -\! S_d(x_k)\! =\! \int_{k\delta}^{(k+1)\delta} \! \! \! \! \! \! \mathrm{L}_{f_d}S_d(x(s))\mathrm{d}s
\end{align}
with $ x(s) =e^{s \mathrm{L}_{f_d}}x|_{x_k}$, admits a unique solution $u_k = \gamma^\delta(x_k)$ in the form of a series expansion in powers of $\delta$ around $\gamma(\cdot)$; i.e. 
\begin{align}\label{asyIH$_d$M}
\gamma^\delta(x) = \gamma(x) + \sum_{i>0}\frac{\delta^i}{(i+1)!} \gamma^i(x)
\end{align}
with suitably defined smooth functions $\gamma_i(\cdot):  \mathbb{R}^n \to \mathbb{R}$.\end{proposition}

\smallskip
The basic idea is to match, at all sampling instant $t = k\delta; k\geq 0$, under piecewise constant control,  the target evolution of the function $S_d(\cdot)$ along the continuous-time closed-loop dynamics (\ref{dyn2}) for $v = 0$, when setting $x_k=x(t=k\delta)$. 
The left and right hand sides of (\ref{IS$_d$M}) define the increment between two successive sampling instants of $S_d(x)$ over the sampled-data dynamics (\ref{sdequiv}) and the continuous-time one (\ref{dyn2})  respectively. Roughly speaking, the feedback $\gamma^\delta(\cdot)$ ensures matching at all sampling instants of the energy consumption along the closed-loop system (\ref{dyn2}).


The IS$_d$M  feedback is implicitly defined by the nonlinear IS$_d$M  equality (\ref{IS$_d$M}) so that seeking for exact solutions might be tough. Still, each term of the series (\ref{asyIH$_d$M}) can be computed through an iterative procedure solving, at each step, a linear equation in the unknown $\gamma^i(x)$. For the first terms, one gets
 \begin{equation}\label{sd_app}
 \begin{split}
& \gamma^1(x) = \mathrm{L}_{f_d}\gamma(x) \\
 & \gamma^2(x) = \mathrm{L}^2_{f_d} \gamma(x) + \frac{\mathrm{L}_{ad_f g} S_d(x)}{2 \mathrm{L}_g S_d(x)}\mathrm{L}_{f_d}\gamma(x) 
 \end{split}
 \end{equation}
 so that as $\delta \to 0$, $\gamma^\delta(x) \to \gamma(x)$ and recovers the continuous-time solution. 

\subsection{The main result}
\smallskip
\begin{theorem}\label{th:main}
Let (\ref{dyn1}) verify Assumption \ref{As1} with storage function $S_d(\cdot): \mathbb{R}^n \to \mathbb{R}_{\geq 0}$ and $\mathrm{L}_g S_d(x)\neq 0$ for all $x \neq x_\star$; then, there exists $T^\star>0$ such that for all $\delta \in [0, T^\star[$ and $k\geq 0$, the digital feedback
\begin{align}\label{feed_sd}
u = \gamma^\delta(x) + v
\end{align}
with $\gamma^\delta(\cdot)$ solution to the IS$_d$M equality (\ref{IS$_d$M}) makes the closed-loop system passive with storage function $S_d(\cdot)$ and passifying output 
\begin{align}\label{out_sd}
h_d^\delta(x,v) = \bar \nabla^\top S_d\big|_{x+F^\delta(x,\gamma^\delta(x))}^{x+F^\delta(x,\gamma^\delta(x)+v)} g^\delta(x,v)
\end{align}
with $vg^\delta(x,v)=F^\delta(x,\gamma^\delta(x)+v)-F^\delta(x,\gamma^\delta(x))$. In addition, if $S_d(x_\star) = 0$ and the continuous-time system with output $y = h_d(x)$ is zero-state detectable, then the feedback (\ref{feed_sd})
with damping part $v = v_\text{di}^\delta(x)$ defined as the solution to the implicit damping equality 
\begin{align}\label{damp}
\delta v + \kappa h_d^\delta(x,v) = 0, \quad \kappa >0
\end{align}
makes $x_\star$ asymptotically stable in closed loop.
\end{theorem}

\begin{proof}
First, let us prove that the dynamics (\ref{sdequiv}) with output (\ref{out_sd}) is made passive by the IS$_d$M control (\ref{feed_sd}). By Assumption \ref{As1}, $\mathrm{L}_{f_d} S_d(x) \leq 0$ and the feedback  $\gamma^\delta(\cdot)$ solution to (\ref{IS$_d$M}) ensures by construction
$
S_d(x\! +\! F^\delta(x, \gamma^\delta(x))) \! -\! S_d(x) \leq 0.
$
As a consequence, exploiting the properties of the discrete gradient function, one gets the dissipation inequality 
\begin{align*}
&S_d(x\! +\! F^\delta(x, \gamma^\delta(x)\! +\! v)) \! -\! S_d(x\!) \!\\
&=S_d(x\! +\! F^\delta(x, \gamma^\delta(x)\! +\! v)) \! -\! S_d(x\! +\! F^\delta(x, \gamma^\delta(x))) \!\\
&+S_d(x\! +\! F^\delta(x, \gamma^\delta(x) )) \! -\! S_d(x\!) \!
\\ \leq & v  \bar \nabla^\top S_d\big|_{x+F^\delta(x,\gamma^\delta(x))}^{x+F^\delta(x,\gamma^\delta(x)+v)} g^\delta(x,v)  = v h_d^\delta(x,v) 
\end{align*} 
and thus the result. The existence of a solution to the damping equality (\ref{damp}) is guaranteed by the Implicit Function Theorem as in first approximation $h_d^\delta(x,v)=\delta h_d(x) +O(\delta ^2)$. Substituting $v= v^\delta_\text{di}(x)$ into the dissipation inequality above and exploiting \eqref{damp}, one gets $\Delta S_d(x)\!  \leq -\kappa \| h_d^\delta(x,v^\delta_\text{di}(x))\|^2$,
so that asymptotic stability of $x_\star$ follows from ZSD in continuous time; see \cite{monaco2011nonlinear} for a detailed proof.
\end{proof}
\smallskip

The stabilizing controller (\ref{feed_sd}) has two components: the passifying feedback $\gamma^\delta(x)$ matching the closed loop storage behaviour and the output damping feedback $v^\delta_\text{di}(x)$, achieving asymptotic stabilization. In practice, as exact solutions might not be computable, only approximate ones are implemented. Let the $p^\text{th}$-order approximate solution be the truncation of the series expansion at any finite order $p$ in $\delta$; i.e.
\begin{align}\label{trunc}
u^\delta_{[p]}(x)\! =\! \gamma(x)\!-\! \kappa h_d(x)\! + \sum_{i = 1}^p \frac{\delta^i}{(i+1)!}u^i(x) 
\end{align}
with $u^i(x) = \gamma^i(x)+v^i_\text{di}(x)$, the correcting term of order $i>0$. 
In particular, setting $p = 0$, one recovers the standard emulated solution or the continuous-time feedback directly implemented through ZOH devices. 
Along the lines of \cite{tanasa2015backstepping}, it can be shown that increasing the approximation order of the controller (i.e., $p>0$) significantly improves the stabilizing performances in closed loop  by reducing the matching error in (\ref{IS$_d$M}) in $O(\delta^{p+2})$.  
\smallskip

\begin{remark}\label{remark:32a} The exact control solution  $u^\delta(x) = \gamma^\delta(x) + v^\delta_\text{di}(x)$ ensures \emph{one-step consistency} \cite{nevsic1999sufficient} in closed-loop.  This induces that the $p^\text{th}$-order approximate controller ensures practical asymptotic stability  or convergency
 to a ball of radius in $O(\delta^{p+1})$, containing the origin (see \cite{tanasa2015backstepping} for details).
\end{remark}

\subsection{The average passifying output map} 

To stress the link with average passivity introduced in \cite{monaco2011nonlinear}, it is sufficient to notice that the defined passifying output (\ref{out_sd}) coincides by construction with the average along the sampled data closed loop dynamics of the composed map 
\begin{align}\label{outpa}
    \frac{\partial S_d(x+F^\delta(x,\gamma^\delta(x)+v))}{\partial v}
\end{align}
that is 
\begin{align}
\delta h_d^\delta(x,v) =\frac{1}{v}\int_0^v \frac{\partial S_d(x+F^\delta(x,\gamma^\delta(x) + w))}{\partial w}dw.
\end{align}
Some easy computations show that 
 \begin{align*} &
 \frac{\partial S_d(x+F^\delta(x,\gamma^\delta(x)+w))}{\partial w}\! \\ & =  \mathrm{L}_{G^\delta (\cdot, \gamma^\delta(x) + w)}S_d(x+F^\delta(x,\gamma^\delta(x)+w)
\end{align*}
with the vector field  $G^\delta (\cdot, u)$ verifying (see \cite{monaco2010sampled} for details)
\begin{align*}
G^\delta (x+F^\delta(x,u), u)= \frac{\partial F^\delta(x,u)}{\partial u}.
\end{align*}
Accordingly, the following result rephrases Theorem \ref{th:main}.

\smallskip
\begin{theorem}
Given a continuous-time feedback passive dynamics (\ref{dyn1}), with storage function $S_d(\cdot): \mathbb{R}^n \to \mathbb{R}_{\geq 0}$ and passifying output $y=\mathrm{L}_g S_d(x)$,  then, there exists $T^\star>0$ such that for all $\delta \in ]0, T^\star[$ and $k\geq 0$, the digital feedback $u = \gamma^\delta(x) + v$
with $\gamma^\delta(\cdot)$ solution to the IS$_d$M equality (\ref{IS$_d$M}) makes the closed-loop system average passive with the same storage function $S_d(x)$ and  passifying output map (\ref{outpa}).\end{theorem}

\smallskip

To better  understand the analogy between the continuous-time and sampled-data results, we note that the continuous-time passifying output $y=\mathrm{L}_g S_d(x)$ is the $u$-derivative of the time derivative of the function $S_d(\cdot)$ along the closed loop continuous-time dynamics $\dot x= f(x)+g(x)\gamma(x)+g(x)v$.
\subsection{Some computational aspects}
Computational details are given to characterize the first terms of the exact solutions around the continuous-time solution. 
The sampled-data passifying output (\ref{out_sd}) can be characterized by its series expansion in powers of $\delta$, computing for the first terms
\begin{align*}
h_d^\delta(x,v)& =\delta h_d(x) + \frac{\delta^2}{2}
(\mathrm{L}_{f_d} + v \mathrm{L}_g)(h_d)(x) \\ & + \frac{\delta^2}{2}
 \nabla^\top S_d(x) (\mathrm{L}_g\mathrm{L}_{f}(x) + \gamma(x) \mathrm{L}_g^2(x)) +  O(\delta^3)
\end{align*}
because 
\begin{align*}\delta g^\delta(x,v)\! &=\! \delta g(x) + \frac{\delta^2}{2}(\mathrm{L}_{f_d} +v\mathrm{L_g} )\mathrm{L}_g (x )\\
&+\frac{\delta^2}{2}( \mathrm{L}_g\mathrm{L}_{f}(x) + \gamma(x) \mathrm{L}_g^2(x)) +  O(\delta^3).
\end{align*}
Analogously, for the damping controller $v = v^\delta_\text{di}(x)$ solution to (\ref{damp}), one gets  
\begin{align*}
v^\delta_\text{di}(x) = -\kappa \mathrm{L}_g S_d(x) + \frac{\delta}{2} v_\text{di}^1(x) +O(\delta^2)
\end{align*}
with \begin{align*}
v^1_\text{di}(x) =& \kappa^2 (\mathrm{L}_{f_d} -\kappa \mathrm{L}_g S_d(x)\mathrm{L}_g)\mathrm{L}_g S_d(x) \\ &- \kappa  \nabla^\top S_d(x) (\mathrm{L}_g\mathrm{L}_{f}(x) + \gamma(x) \mathrm{L}_g^2(x)).
\end{align*}

\section{The case of port-Hamiltonian systems}\label{sec:pcH}
Let us now revisit the result in \cite{tiefensee2010ida} in the proposed passivating framework making reference to the discrete gradient function. Let (\ref{dyn1}) be a continuous-time pcH dynamics with
\begin{align}\label{pcS_df}
\dot x= f(x) = (J(x)-R(x))\nabla H(x)
\end{align}
for $J(x) + J^\top(x) = 0$, $R(x) = R^\top(x) \succeq 0$ and smooth Hamiltonian function $H(\cdot) : \mathbb{R}^n \to \mathbb{R}$. In this context, stabilization via passivation results in IDA-PBC design \cite{ortega2002interconnection}. Basically,  one seeks for a control of the form (\ref{passfeed}) to stabilize a desired equilibrium while preserving the pcH structure
so transforming (\ref{dyn2}) into
\begin{align}\label{pcS_df2}
\dot x= f_d(x) = (J_d(x)-R_d(x))\nabla H_d(x)
\end{align}
 that is passive with respect to the output $y_d(x) = g^\top (x) \nabla H_d(x)$ with target Hamiltonian function $H_d(\cdot)$.

\medskip 

From Theorem \ref{th:main}, the proposed sampled-data feedback $\gamma^\delta(\cdot)$ solution to (\ref{IS$_d$M}), assigns in closed loop the target energy function $H_d(\cdot) $ with equilibrium in $x_\star$ as already discussed in \cite{tiefensee2010ida}. However, it can be easily verified that the IH$_d$M based digital feedback does not  preserve the port-Hamiltonian structure in closed loop in general so that it cannot be properly referred as a sampled-data IDA-PBC design. More precisely, the feedback (\ref{feed_sd}) with $v=0$ does not assign to the exact sampled-data model a discrete-time pcH structure along the definition adopted in the literature (e.g. \cite{yalccin2015discrete,aoues2017Hamiltonian,moreschini2019discrete,moreschini2020stabilization}) \begin{align}\label{pcS_dt}
x_{k+1} = x_k + \delta(J_d^\delta(x_k) -R^\delta_d(x_k))\bar \nabla H_d|_{x_k}^{x_k + F_d^\delta(x_k)} 
\end{align}
with $J_d^\delta(x) + \big(J_d^\delta(x) \big)^\top = 0$ and $R^\delta_d(x) = \big(R^\delta_d(x)\big)^\top \succeq 0$. 

The following result shows that the $1^\text{st}$-order approximate feedback
\begin{align*}
\gamma^\delta_{[1]}(x) = \gamma(x) + \frac{\delta}{2} \mathrm{L}_{f_d}(\gamma)(x)
\end{align*}
ensures (local) passivation and recovers in closed loop a  pH structure up to an error in $O(\delta ^3)$.
\smallskip
\begin{theorem}
\label{th_2}
Given a continuous-time pcH dynamics (\ref{pcS_df}) and an IDA-PBC feedback assigning the dynamics (\ref{pcS_df2}) for given matrices $J_d(x)$, $R_d(x)$ and Hamiltonian function $H_d(x)$, let $u_k = \gamma^\delta(x_k)$, be the sampled-data feedback solution to the IH$_d$M equality (\ref{IS$_d$M}). Then, the  closed-loop sampled-data dynamics exhibits in $O(\delta ^3)$ a pcH structure 
\begin{align}\label{eq_1}
 x_{k+1}&=\delta (J_d^\delta(x) - R_d^\delta(x))\bar \nabla H_d|_{x}^{x + F_{d, [2]}^\delta(x)} 
\end{align} 
with \begin{subequations}
\begin{align}
F_{d,[2]}^\delta(x)&=\delta (I_d+ \frac{\delta}{2} J_x[ f_d(x)])M_d(x) \nabla H_d(x)\\
J_d^\delta(x) &= \frac{1}{2}\big(M_d^\delta(x) - (M_d^\delta(x))^\top \big)\label{J}
\\ 
R_d^\delta(x) &= -\frac{1}{2}\big(M_d^\delta(x) + (M_d^\delta(x))^\top \big)\label{Rr}\\
\label{Mdde}
M_d^\delta(x) &= (I_d + \frac{\delta}{2} J_x [f_d(x)]) M_d(x)\times\\ & \big[ I_d +\frac{\delta}{2}\nabla^2 H_d(x) (I_d + \frac{\delta}{2} J_x[f_d(x)])M_d(x)\big]^{-1} \nonumber
\end{align}
\end{subequations}
with $M_d(x) = J_d(x)-R_d(x)$ and $\mathrm{J}_x[\cdot]$, the Jacobian matrix of the function into the brackets. \end{theorem}

\smallskip
\begin{proof}
By definition,  the approximation in $O(\delta^3)$ of  the exact sampled-data equivalent to (\ref{pcS_dt}) for $v=0$ gives
\begin{align*}
x_{k+1} &= x_k +F_{d,[2]}^{\delta}(x_k)=f_d(x_k)+\mathrm{L}_{f_d}^2(x_k)
\end{align*}
with $ f_d(x) = M_d(x) \nabla H_d(x)$ and
$ \mathrm{L}_{f_d}^2(x)=\big( \mathrm{J}_x [f_d(x)] \big) M_d(x) \nabla H_d(x)$.
On the other side, from (\ref{discrete:approx}), one computes
\begin{align*}
 \bar \nabla H_d|_{x}^{x + F_{d, [2]}^\delta(x)}  =& 
Q_d^\delta(x) \nabla H_d(x) + O(\delta^3).
\end{align*}
with $Q_{d}^\delta(x) = \text{I} + \frac{\delta}{2}\nabla^2 H_d(x) (\text{I} + \frac{\delta}{2} J_x [f_d(x)])M_d(x)$.
Substituting these expressions into (\ref{eq_1}), one easily gets that it is solved in $O(\delta^3)$ by setting
$M_d^\delta(x)$ as in (\ref{Mdde}) and (\ref{J})-(\ref{Rr}). Accordingly, because $R_d^\delta(x) = R_d(x) + \delta \tilde R_d^\delta(x)$ one gets that $R_d^\delta(x) \succeq 0$ for $\delta$ small enough and thus the result.
\end{proof}


\section{The gravity pendulum as an example} \label{sec:ex}
Consider the damped pendulum actuated by the torque $u$ and described in port-Hamiltonian form by the equation
\begin{align}\label{pend}
\begin{pmatrix}
\dot{q}\\\dot{p}
\end{pmatrix}&=\begin{pmatrix}
0 & 1\\-1 & -r
\end{pmatrix}\nabla H(q, p)+ \begin{pmatrix}
0\\1
\end{pmatrix}u
\end{align}
with $q$ and $p$ the angular displacement from the vertical axes and velocity, $H(x)=\frac{1}{2}p^2 + (1-\cos(q))$ the Hamiltonian function, $r>0$ the viscous damping coefficient. Setting $x = (q, \ p)^\top$ and $x_\star = (q_\star, 0)^\top$ the desired equilibrium, \eqref{pend} satisfies Assumption \ref{As1} with  $u = \gamma(x) + v$ and
\begin{align}\label{contu}
\gamma(x)= \sin(q) - \sin(q-q_\star),
\end{align}
	making the closed loop
	\begin{align*}
	\begin{pmatrix}
	\dot{q}\\\dot{p}
	\end{pmatrix}&=\begin{pmatrix}
	0 & 1\\-1 & -r
	\end{pmatrix}\nabla H_d(q, p) + \begin{pmatrix}
	0\\1
	\end{pmatrix}v, \quad y = p
	\end{align*}
	passive with desired storage function
	\begin{align}\label{desired}
	H_d(x)=\frac{1}{2}p^2 + (1-\cos(q-q_\star)).
	\end{align}
	The Hamiltonian function ${H}_d(x)$ along the closed-loop dynamics verifies
$
\dot{H}_d(x)= -rp^2 + vy.
$
As discussed in Section \ref{sec:main}, the sampled equivalent model (\ref{sdequiv}) is approximated in $O(\delta^3)$ as 
\begin{align}\label{sd_dyn}
&\begin{pmatrix}
q_{k+1}\\p_{k+1}
\end{pmatrix}=\begin{pmatrix}
q_{k}\\p_{k}
\end{pmatrix}+\delta\begin{pmatrix}
p_k \\-\sin(q_k) -rp_k
\end{pmatrix}\\ 
&+ \frac{\delta^2}{2} \begin{pmatrix}
              - \sin(q_k) - rp_k\\
r(\sin(q_k) + rp_k) - p_k\cos(q_k)
\end{pmatrix} \nonumber + \delta\begin{pmatrix}
\frac{\delta}{2}\\1 -\frac{\delta}{2}r
\end{pmatrix}  u_k. \nonumber
\end{align} 
One computes the $1^{st}$-order approximated feedback  (\ref{trunc}) 
\begin{align}\label{control2}
u_k=\gamma(x_k) + \frac{\delta}{2}\gamma_1(x_k) + v_k, 
\end{align}
with
$\gamma_1(x)= (\cos(q) - \cos(q \!-\! q_\star))p$. It follows that the sampled-data dynamics is passive with storage function \eqref{desired} and passifying output given by
\begin{align*}
	h_d^\delta(x_k,v_k)= \delta p_k + \frac{\delta^2}{2}(v_k - 2rp_k) + O(\delta^3).
\end{align*} 
Accordingly, the closed-loop system  under \eqref{control2} verifies the dissipation inequality
\begin{align*}
H(x_{k+1}) - H(x_k) \leq \delta p_kv_k + \frac{\delta^2}{2}(v^2_k - 2rp_kv_k) + O(\delta^3).
\end{align*} 
In addition, from Theorem \ref{th:main} the implicit damping equality
\begin{align*}
\delta v_k =- \delta \kappa p_k - \frac{\delta^2}{2}(\kappa v_k - 2\kappa rp_k) + O(\delta^3),\  \kappa >0
\end{align*}
can be solved  in $ O(\delta^3)$ with
\begin{align*}
v_{\text{di}}(x) = -  \kappa p + \frac{\delta}{2} \kappa (2r + \kappa) p + O(\delta^2)
\end{align*}
which ensures convergence to a ball containing $x_\star$ of radius in $O(\delta^{2})$.
From Theorem \ref{th_2},  under the feedback \eqref{control2} with $v_k=0$, the dynamics \eqref{sd_dyn} gets the port-Hamiltonian form \eqref{pcS_dt} in  $O(\delta^{3})$; i.e. (\ref{eq_1}) is satisfied with
\begin{align*}
F_{d,[2]}^\delta(x)=&\delta\!\begin{pmatrix}
p \\-\sin(q\!-\! q_\star) -rp
\end{pmatrix}\\
&+\! \frac{\delta^2}{2}\!\begin{pmatrix}
-\sin(q\!-\! q_\star) -rp\\r\sin(q \!-\! q_\star) +p (r^2 - \cos(q \!-\! q_\star))
\end{pmatrix}\\ 
	M_d^\delta(x) =& \begin{pmatrix}
	0 & 1 \\ -1& -r
	\end{pmatrix}, 	J_d^\delta = \begin{pmatrix}
	0 & 1 \\ -1& 0
	\end{pmatrix}, 	R_d^\delta = \begin{pmatrix}
	0 & 0 \\ 0 & r
	\end{pmatrix}.
\end{align*}
Setting $\nu = (\nu_1 \ \nu_2)^\top$, $\mu = (\mu_1 \ \mu_2)^\top$, the discrete gradient
\begin{align*}
\bar\nabla H_d|_{\nu}^{\mu}= \begin{pmatrix}
-\frac{\cos(\mu_1-q_\star)-\cos(\nu_2-q_\star)}{\mu_1-\nu_1}
\\ 
\frac{1}{2} (\mu_2+\nu_2)
\end{pmatrix}
\end{align*} 
 along the closed loop dynamics is approximated in $O(\delta^{2})$ as 
\begin{align*}
\bar\nabla H_d|_{x}^{x + F_{d,[2]}^\delta(x)}=&\begin{pmatrix}
\sin(q\!-\!q_\star)\\ p
\end{pmatrix} + 
\frac{\delta}{2}\begin{pmatrix}
\cos(q\!-\!q_\star)p \\ -\sin(q\!-\!q_\star) - rp
\end{pmatrix}.
\end{align*}

\emph{Simulations.} Setting $x_\star=col(\frac{\pi}{2}, 0)$, initial condition $x_0=col(0, 0)$, and $r=0.4$, simulations are reported. Fig.\ref{fig:4b} compares the effect of the continuous-time feedback \eqref{contu}, $0$-order approximated feedback (emulation of \eqref{contu}) and $1^\text{st}$-order approximate control in \eqref{control2} applied to the pendulum system \eqref{pend}. The purpose is to show their performances over the desired Hamiltonian function \eqref{desired} and the closed-loop trajectories for $\delta=1$ s. Fig.\ref{fig:4b} highlights that the proposed solution significantly approaches both the continuous-time trajectories and the Hamiltonian $H_d(\cdot)$ while the emulated feedback suffers the step-size of $\delta=1s$ and is not able to match neither the desired Hamiltonian nor to stabilize the desired equilibrium. Fig.\ref{fig:3b} compares the Root Mean Squared Error (RMSE) in $H_d(\cdot)$ under approximate solutions (\ref{trunc}) for $\delta\in[0.05: 0.05: 1.5]$, $p = 1$, $p = 0$ respectively with $v = 0$. Improvement is clear even for small values of $\delta$.

\begin{figure}[t!]
	\centering
	\includegraphics[width=.45\textwidth]{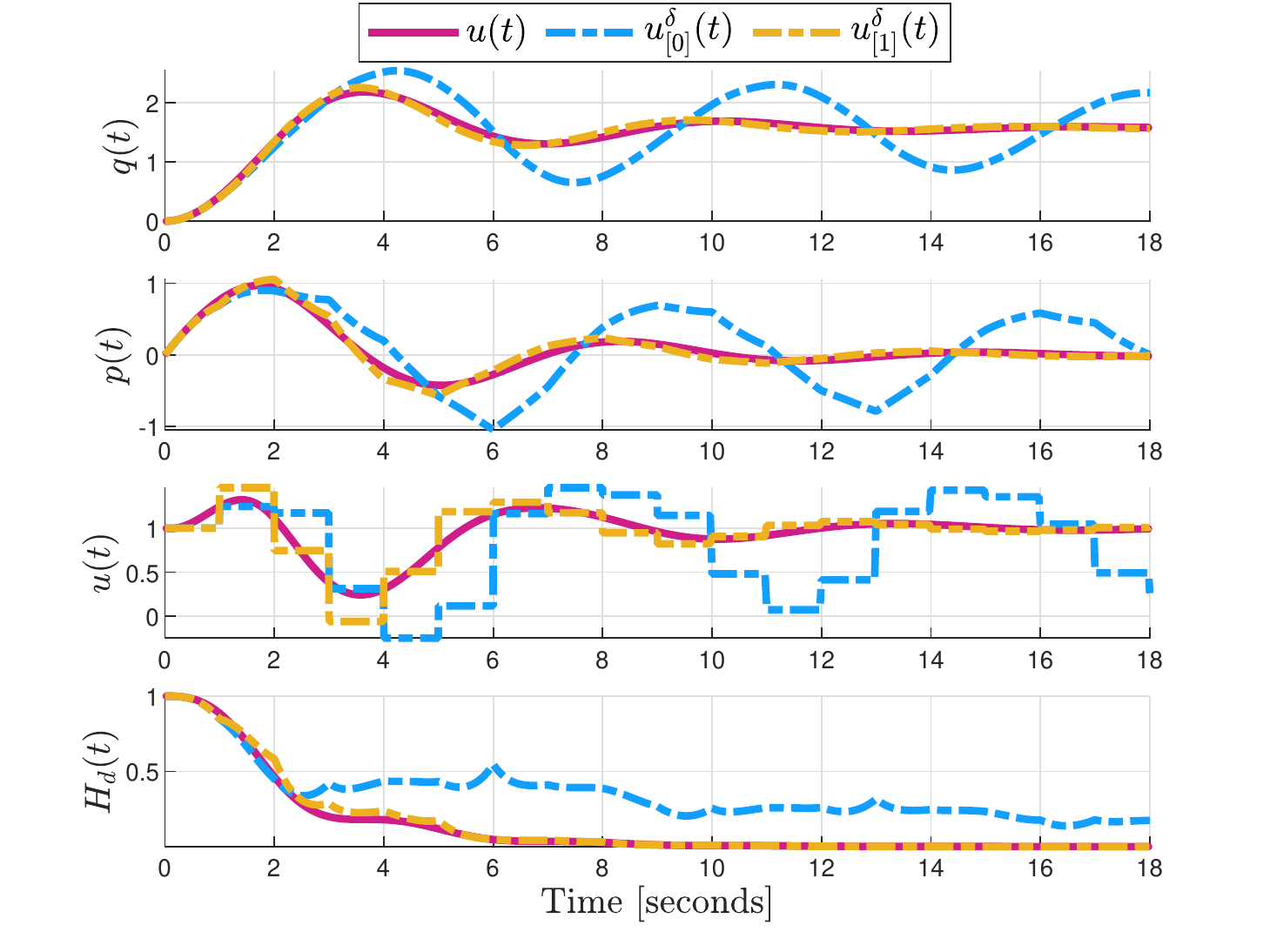}  
	\caption{Continuous-time design $\gamma(x)$, emulated feedback $u^{\delta}_{[0]}(x)$, and the $1$-order approximated feedback $u^{\delta}_{[1]}(x)$, for $\delta=1$ and $\kappa = 0.1$.}
	\label{fig:4b}
\end{figure}
\begin{figure}[t!]
	\centering
	\includegraphics[width=.38\textwidth]{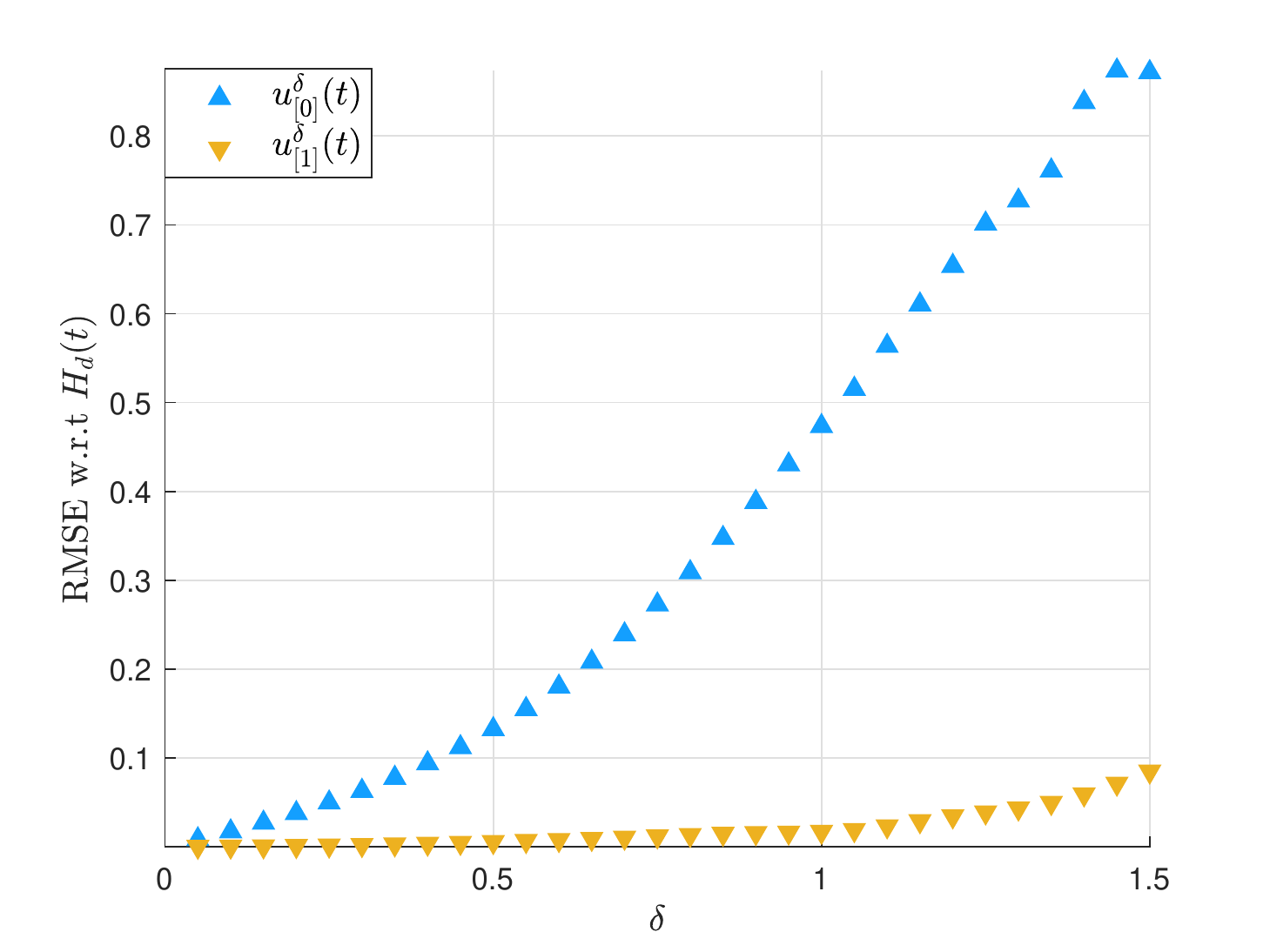}  
	\caption{Matching error for $\delta\in[0.05, 1.5]$.}
	\label{fig:3b}
\end{figure} 

\section{Conclusions and Perspectives} \label{sec:conc}
It has been shown that feedback passivation can be preserved under sampled-data control with respect to the same target storage as in continuous time and suitably modified output mapping. An interpretation of the modified output mapping is given in terms of averaging. The result is appealing for systems that may not be passive in open loop. New perspectives concern revisiting under sampling and in this unifying framework  all continuous-time time design strategies involving feedback passivation as an instrumental tool as for instance for networked systems and energy transfer management among the different ports.
When applied to pcH dynamics, damping assignment is achieved but a sampled-data port Hamiltonian structure is recovered in $O(\delta^3)$ only. Further work is toward generalizing this result according to a deeper understanding of the geometric structure behind  discrete-time pcH forms. 
\bibliographystyle{IEEEtran}      
\bibliography{biblio}                  

\end{document}